\documentclass[journal]{IEEEtran}

\usepackage{amsmath}
\usepackage{amsfonts}
\usepackage{amssymb}
\usepackage{amsthm}
\usepackage{tabularx}
\usepackage{mathrsfs}
\usepackage{graphicx}
\usepackage{subcaption}
\usepackage{caption}
\usepackage{url}
\usepackage{hyperref}

\newtheorem{remark}{Remark}

\newtheorem{Proposition}{Proposition}

\usepackage{stfloats}
\usepackage{float}
\usepackage{graphicx}
\usepackage{xcolor}

\makeatletter
\def\blfootnote{\xdef\@thefnmark{}\@footnotetext}
\makeatother

\begin{document}

\title{\LARGE UAV-Relay Assisted RSMA Fluid Antenna System:\\ Outage Probability Analysis } 
\author{Farshad Rostami Ghadi, \IEEEmembership{Member}, \textit{IEEE}, 
            Masoud Kaveh, 
            Francisco Hernando-Gallego,  
            Diego Mart\'in,\\ 
	    Kai-Kit Wong, \IEEEmembership{Fellow}, \textit{IEEE}, and 
	    Chan-Byoung Chae, \IEEEmembership{Fellow}, \textit{IEEE}
\vspace{-9mm}
}
\maketitle

\blfootnote{The work of F. Rostami Ghadi is supported by the European Union's Horizon 2022 Research and Innovation Programme under Marie Sk\l odowska-Curie Grant No. 101107993. The work of K. K. Wong is supported by the Engineering and Physical Sciences Research Council (EPSRC) under Grant EP/W026813/1. The work of M. Kaveh is supported in part by the Academy of Finland under Grants 345072 and 350464. The work of C. B. Chae is supported by the Institute for Information and Communication Technology Planning and Evaluation (IITP)/NRF grant funded by the Ministry of Science and ICT (MSIT), South Korea, under Grant RS-2024-00428780 and 2022R1A5A1027646.}

\begin{abstract}
This letter studies the impact of fluid antenna system (FAS) technology on the performance of unmanned aerial vehicle (UAV)-assisted multiuser communication networks. Specifically, we consider a scenario where a fixed-position antenna (FPA) base station (BS) serves $K$ FAS-equipped users with the assistance of a UAV acting as an aerial relay. The BS employs rate-splitting multiple access (RSMA), while the UAV operates in half-duplex (HD) mode using the decode-and-forward (DF) strategy. For this system, we derive a compact analytical expression for the outage probability (OP) and its asymptotic behavior in the high signal-to-noise ratio (SNR) regime, leveraging the multivariate $t$-distribution. Our results show how deploying FAS at ground users (GUs) in UAV-aided communications improves overall system performance compared to using FPA GUs.
\end{abstract}


\maketitle
	
\blfootnote{\noindent F. Rostami Ghadi and is with the Department of Signal Theory, Networking and Communications, Research Centre for Information and Communication Technologies (CITIC-UGR), University of Granada, 18071, Granada, Spain (e-mail: $\rm f.rostami@ugr.es)$.}
\blfootnote{\noindent M. Kaveh is with the Department of Information and Communication Engineering, Aalto University, Espoo, Finland (e-mail: $\rm masoud.kaveh@aalto.fi$).}
\blfootnote{\noindent F. Hernando-Gallego and D. Mart\'in are with the Department of Computer Science, Escuela de Ingenier\'{i}a Inform\'{a}tica de Segovia, Universidad de Valladolid, Segovia, 40005, Spain (e-mail: $\rm \{diego.martin.andres, fhernando\}@uva.es$).}
\blfootnote{\noindent K. K. Wong is affiliated with the Department of Electronic and Electrical Engineering, University College London, Torrington Place, WC1E 7JE, United Kingdom and also affiliated with Yonsei Frontier Lab, Yonsei University, Seoul, Republic of Korea (e-mail: $\rm kai\text{-}kit.wong@ucl.ac.uk$).}
\blfootnote{\noindent C. B. Chae is with School of Integrated Technology, Yonsei University, Seoul, 03722, Republic of Korea. (e-mail: $\rm cbchae@yonsei.ac.kr)$.}

\blfootnote{Corresponding author: Francisco Hernando-Gallego.}

\vspace{-3mm}
\section{Introduction}\label{sec-intro}
\IEEEPARstart{F}{luid} antenna system (FAS) has emerged as a transformative technology that dynamically adjusts its antenna location to mitigate small-scale fading and enhance wireless communication reliability \cite{Wong2020FAS}. Using spatial flexibility, FAS improves signal quality and interference management, making it a promising candidate for next-generation networks (NGNs), a.k.a., sixth-generation (6G) systems \cite{wong2022bruce,new2024tutorial}. The integration of FAS with advanced technologies such as reconfigurable intelligent surfaces (RIS) \cite{Ghadi2024RISFAS,shoja2022mimo,Zhu2025fluid,rostami2025phase}, integrated sensing and communication (ISAC) \cite{zhou2024fluid,rostami2024isac,zou2024shift}, non-orthogonal multiple access (NOMA) \cite{new2024noma}, and rate-splitting multiple access (RSMA) \cite{Rostami2024rsma} has been extensively investigated, demonstrating its potential to enhance spectral and energy efficiency.

FAS could be synergized with other advanced communication paradigms, such as unmanned aerial vehicle (UAV)-assisted communication. UAVs have recently attracted significant attention in wireless networks due to their ability to serve as mobile aerial relays, providing enhanced coverage, adaptability, and reliability, especially in challenging environments \cite{Li2021UAVRelaying}. In this regard, considering the UAV as a relay could be particularly valuable when direct transmission between base stations (BSs) and ground users (GUs) is obstructed or suffers from severe path loss. As 6G aims to meet the ultra-reliable low-latency communication (URLLC) and seamless connectivity requirements, UAV relays could play a key role in tackling these challenges. However, the performance of UAV relaying is often limited by unpredictable channel variations and dynamic mobility constraints. In this context, FAS could help by enabling antennas on the user side to dynamically adjust in real-time, optimizing signal reception and improving network reliability. Only recently, \cite{shen2025ris,abdou2024sum,Meng-2025} applied FAS to UAV and UAV-relay systems, respectively, formulating a sum-rate maximization problem. Their results demonstrated that FAS improves sum-rate performance compared to deploying fixed-position antennas (FPAs). Therefore, the combination of FAS and UAVs presents an intriguing possibility, where UAVs offer mobility and flexible deployment, while FAS could dynamically adjust antenna positions to optimize signal reception based on time-varying channel conditions. This synergy could enable UAVs to better combat fading effects, improve signal reception, and enhance overall transmission reliability, making it a promising avenue for 6G networks.

Motivated by the above, this work explores the integration of FAS and UAV-assisted relaying into a unified multiuser communication framework. Specifically, we consider a half-duplex (HD) UAV relay using the decode-and-forward (DF) strategy to assist a FPA BS in serving multiple FAS-equipped GUs. Instead of considering NOMA which suffers from inter-user interference due to successive interference cancellation (SIC) complexities, the BS employs RSMA signaling which mitigates interference by splitting user messages into common and private parts \cite{Mao2022RSMA}. For this model, we derive a compact analytical expression for the outage probability (OP) and its asymptotic expression at high signal-to-noise ratio (SNR) using the multivariate $t$-distribution. Our results demonstrate the great performance gains by deploying FAS at GUs in the UAV-relayed RSMA system, outperforming conventional FPA systems and multiple access schemes in terms of reliability, interference management, and adaptability.

\vspace{-3mm}
\section{System Model}\label{sec-sys}
\subsection{Network Model}
We consider a UAV-assisted communication system, as shown in Fig.~\ref{fig_model}, where a BS serves $K$ GUs via an aerial relay, i.e., UAV. For simplicity, the indices $\{b, a, k\}$ represent the BS, UAV, and $k$-th GUs, respectively. It is assumed that direct communication links between the BS and the GUs are blocked due to obstacles and significant link attenuation. Hence, the aerial relay operates in HD mode, employing a DF strategy to facilitate communication between the BS and the GUs. The UAV is assumed to hover at a fixed location, with coordinates $\mathbf{r}_{a} = (x_a, y_a, z_a)$,\footnote{The OP performance can be determined for any specific UAV location, allowing for straightforward extension to scenarios involving a moving UAV, though such cases are omitted here for brevity.} while the BS and GUs are located at the origin of the Cartesian coordinate system, $\mathbf{r}_b = (0, 0, 0)$, and at $\mathbf{r}_{k} = (x_{k}, y_{k}, z_{k})$ for $k = 1, \dots, K$, respectively. The BS and UAV are equipped with single FPAs, while each GU features a grid structure consisting of $N_k^l$ ports, uniformly distributed over two linear spaces of length $W_k^l\lambda$ for $l\in\left\{1,2\right\}$, with $\lambda$ being the wavelength. Thus, the total number of ports for the $k$-th GU is $N_{k}=N_{k}^1\times N_{k}^2$ and the total area of the  surface is $W_{k}=\left(W_k^1\times W_k^2\right)\lambda^2$. Additionally, to simplify the indexing of these 
two-dimensional (2D) structures, we introduce a mapping function $\mathcal{F}\left(n_{k}\right)=\left(n_k^1,n_k^2\right)$, $n_{k}=\mathcal{F}^{-1}\left(n_k^1,n_k^2\right)$, which converts the 2D indices into a one-dimensional (1D) form such that $n_{k}\in\left\{1,\dots,N_{k}\right\}$ and $n_k^l\in\left\{1,\dots,N_k^l\right\}$.

\begin{figure}[!t]
\centering \includegraphics[width=0.7\columnwidth]{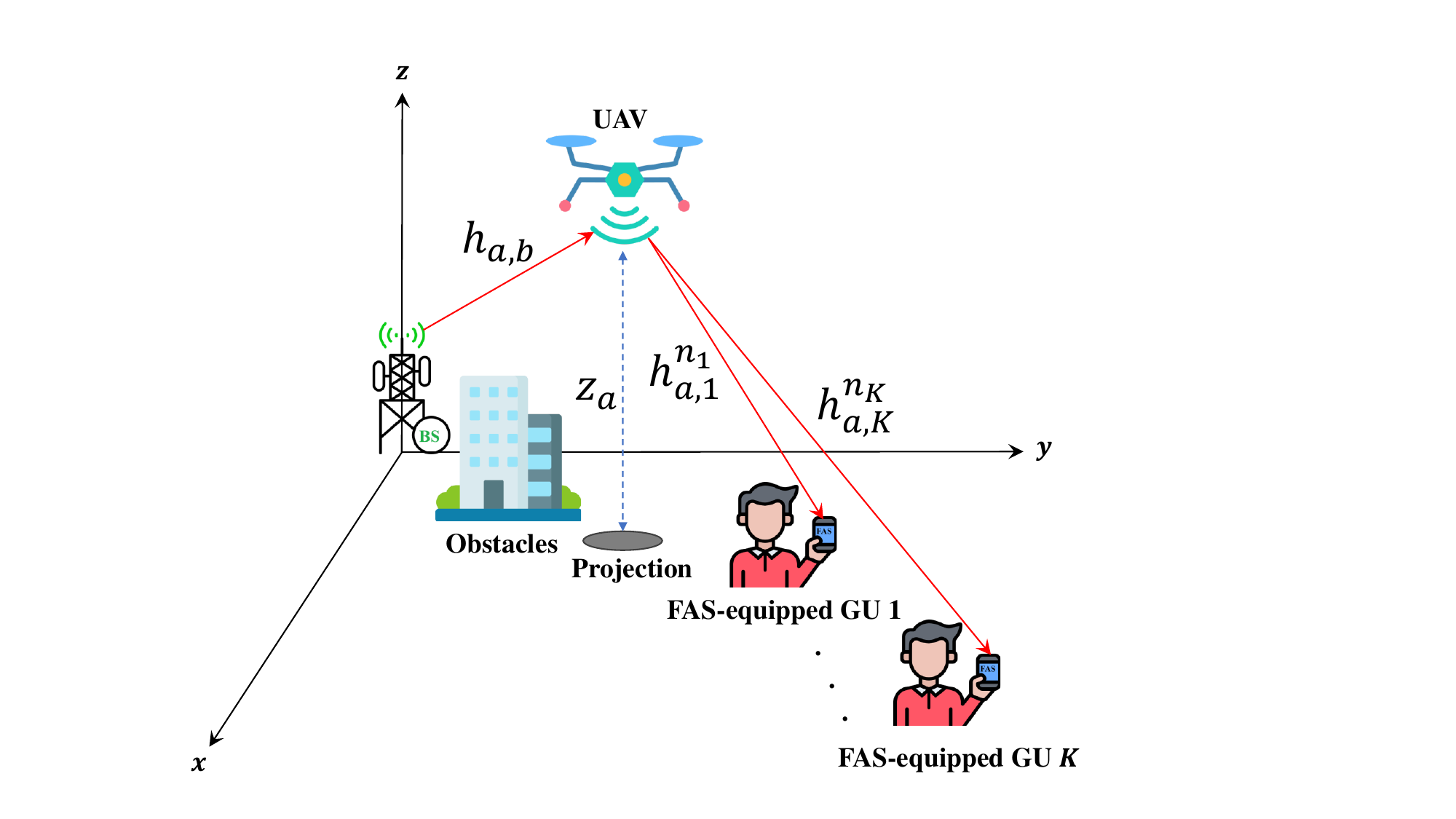}
\caption{Sketch of the UAV-relay assisted FAS.}\vspace{0cm}\label{fig_model}
\vspace{-3mm}
\end{figure}

\vspace{-2mm}
\subsection{Channel Model}
The air-ground channels, i.e., the channels between the BS and UAV and the channels between the UAV and each GU, are modeled with  probabilistic line-of-sight (LoS) components and non-line-of-sight (NLoS) components. According to \cite{hourani2014optimal}, the probability of LoS for an air-ground link with a UAV and a ground node $i\in\left\{b,k\right\}$ is defined as
\begin{align}
P_{i}^\mathrm{LoS}=\frac{1}{1+\mu_1\mathrm{e}^{-\mu_2\left(\theta_i-\mu_1\right)}},
\end{align}
where $\mu_1$ and $\mu_2$ are the constant parameters related to the  propagation environment \cite{hourani2014optimal}. Besides, $\theta_i$ is the corresponding elevation angle, denoted as
\begin{align}
\theta_i = \mathrm{tan}^{-1}\left(\frac{z_a}{d_{a,i}}\right),
\end{align}
where $d_{a,i}=\|\mathrm{r}_a-\mathrm{r}_i\|$ defines the distance between the ground node $i$ and the UAV's projection on the ground (i.e., the horizontal distance), which can be computed as
\begin{align}
d_{ai} = \sqrt{\left(x_a-x_i\right)^2+\left(y_a-y_i\right)^2}.
\end{align}
As a result, the probability of NLoS can be defined as
\begin{align}
P_i^\mathrm{NLoS} = 1-P_i^\mathrm{LoS}.
\end{align}
Therefore, the corresponding path-loss coefficient of the air-ground channel can be expressed as
\begin{align}\label{eq-path}
L_{a,i}= \left(\eta_1P_i^\mathrm{LoS}+\eta_2P_i^\mathrm{NLoS}\right)r_{a,i}^{-\beta}, 
\end{align}
in which $\eta_1$ and $\eta_2$ are the reference path-loss coefficients at the distance of 1 meter for the LoS and NLoS links, respectively. The term $\beta$ denotes the path-loss exponent and $r_{a,i}$ is the distance of the air-ground link between the UAV and node $i$, given by
\begin{align}
r_{a,i}=\sqrt{{d_{a,i}^2+\left(z_a-z_i\right)^2}}.
\end{align}

Based on the channel measurement presented in \cite{yanmaz2013achieving}, the small-scale fading channel of the air-ground links between the UAV and node $i$, denoted as $h_{a,i}$, can be characterized using Nakagami-$m$ distribution. Hence, the corresponding air-ground channel gain $g_{a,i}=|h_{a,i}|^2$ is modeled as Gamma distribution with the cumulative distribution function (CDF)
\begin{align}\label{eq-cdf1}
F_{g_{a,i}}\left(g\right) = \frac{\Upsilon\left(m_i,\frac{m_i}{\Omega_i}g\right)}{\Gamma\left(m_i\right)}, 
\end{align}
where $\Gamma\left(a\right)$ and  $\Upsilon\left(a,b\right)$ are Gamma function and the lower incomplete Gamma function, respectively. Besides, $\Omega_i$ and $m_i$ are the spread and shape parameters, respectively. However, considering the simplest FAS, only the port that provides the highest received signal-to-interference-plus-noise ratio (SINR) for FAS-equipped GUs is activated. Consequently, the received channel gain at the $k$-th GU is expressed as
\begin{align}\label{eq-g-fas}
g_{\mathrm{fas},k} = \max\left\{g_{a,k}^{1},\dots,g_{a,K}^{N_K}\right\}=g_{a,k}^{n_k^*},
\end{align}
where $n_{k}^*$ denotes the best port index at the $k$-th GU that maximizes the channel gain, i.e.,
\begin{align}
n_{k}^* = \arg\underset{n_k}{\max}\left\{g_{a,k}^{n_k}\right\}.
\end{align}
Moreover, due to the ability of fluid antenna ports to switch positions freely and be placed arbitrarily close to one another, the resulting channels exhibit spatial correlation. In particular, for a given setup with rich scattering, the covariance between any two arbitrary ports, denoted as $n_k$ and $\tilde{n}_k$, at the $k$-th GU is mathematically expressed as \cite{Ghadi2024RISFAS}
\begin{align}
\varrho_{n_k,\tilde{n}_k}^{k}=\mathcal{J}_0\left(2\pi\sqrt{\left(\frac{n_k^1-\tilde{n}_k^1}{N^1_k-1}W_k^1\right)^2+\left(\frac{n_k^2-\tilde{n}_k^2}{N^2_k-1}W_k^2\right)^2}\right),
\end{align}
where $\mathcal{J}_0\left(\cdot\right)$ defines the zero-order spherical Bessel function of the first kind. Therefore, following the results in \cite{Rostami2024rsma}, the CDF of the channel gain at the $k$-th GU (see \eqref{eq-g-fas}), which  includes the maximum of $N_k$ correlated Gamma random variables, can be expressed as
\begin{multline}
	F_{g_{\mathrm{fas},k}}\left(g\right) = 
	T_{\nu_{k},\mathbf{\Sigma}_{k}} \Bigg( t_{\nu_{k}}^{-1} 
	\left( \frac{\Upsilon\left(m_{k},\frac{m_k}{\Omega_k}g\right)}
	{\Gamma\left(m_k\right)} \right), \dots, \\  
	t_{\nu_k}^{-1} \left( \frac{\Upsilon\left(m_k,\frac{m_k}{\Omega_k}g\right)}
	{\Gamma\left(m_k\right)} \right) ; \nu_k, \Theta_k \Bigg), \label{eq-cdf2}
\end{multline}
where $t_{\nu_k}^{-1}(\cdot)$ represents the inverse CDF (quantile function) of the univariate $t$-distribution with $\nu_k$ degrees of freedom for the $k$-th GU. The CDF of the multivariate $t$-distribution, which has a correlation matrix $\mathbf{\Sigma}_k$ and $\nu_k$ degrees of freedom for the $k$-th GU, is denoted by $T_{\nu_k,\mathbf{\Sigma}_k}(\cdot)$. Additionally, the dependence parameter of the Student's $t$-copula, represented by $\Theta_k \in \left[-1,1\right]$, quantifies the correlation between the $n_k$-th and $\tilde{n}_k$-th ports within $\mathbf{\Sigma}_k$, i.e., $\theta_k\approx \varrho_{n_k, \tilde{n}_k}^k$ \cite{Rostami2024rsma}. Also, $\left|\mathbf{\Sigma}_k\right|$ denotes the determinant of the correlation matrix $\mathbf{\Sigma}_k$ and
\begin{align}
\mathbf{t}^{-1}_{\nu_k} = \left[t^{-1}_{\nu_k}\left(\frac{\Upsilon\left(m_k,\frac{m_k}{\Omega_k}g\right)}{\Gamma\left(m_k\right)}\right),\dots,t^{-1}_{\nu_k}\left(\frac{\Upsilon\left(m_k,\frac{m_k}{\Omega_k}g\right)}{\Gamma\left(m_k\right)}\right)\right].
\end{align}

\subsection{Signal Model}
In the first phase, the BS implements RSMA by splitting each GU's message $w_{k}$ into a common component $w_c$ and a private component $w_{p,k}$. The common message $w_c$ is intended for each GU and encoded into a common stream $s_c$ using a shared codebook available to all GUs. This stream is designed to be decoded by every GU before extracting their specific information. Meanwhile, the private messages $w_{p,k}$ are separately encoded into distinct private streams $s_{p,k}$, ensuring that each GU receives its designated data. Hence, the transmitted signal by the BS is expressed as
\begin{align}	x=\underset{\text{common message}}{\underbrace{\sqrt{\alpha_{c}}s_{c}}}+\underset{\text{private message}}{\underbrace{\sum_{k=1}^K\sqrt{\alpha_{{p},k}}s_{{p},k}}},
\end{align}
in which $\alpha_c$ and $\alpha_{p,k}$ are the power allocation factors of $s_c$ and $s_{p,k}$, respectively, that must satisfy  $\alpha_{c}+\sum_{k=1}^K\alpha_{{p},k}=1$. Accordingly, the signal received by the UAV is given by
\begin{align}
y_a = \sqrt{L_{a,b}P_b}h_{a,b}x+z_a, 
\end{align}
where $P_b$ is the power transmitted by the BS and $z_a$ denotes the complex additive white Gaussian noise (AWGN) at the UAV with zero mean and variance of $\sigma^2_a$.

At the UAV, the common stream $s_c$ is  decoded by treating all the private GU signals as noise. Therefore, the instantaneous SINR at the UAV for decoding the common message is 
\begin{align}\label{eq-ac}
\gamma_{a}^{c,k} =\frac{P_b\alpha_cL_{a,b}g_{a,b}}{P_bL_{a,b}g_{a,b}\sum_{k=1}^{K}\alpha_{p,k}+\sigma^2_a}. 
\end{align}
Once $s_c$ is successfully decoded, SIC is applied, and the decoded $s_c$ is subtracted from the received signal. After this, the private signal $ s_{p,k}$ for GU $k$ is decoded, with the private signals from all other GUs treated as noise. Consequently, the SINR for receiving $s_{p,k}$ at the UAV can be found as
\begin{align}\label{eq-ap}
\gamma_{a}^{p,k} =\frac{P_b\alpha_{p,k}L_{a,b}g_{a,b}}{P_bL_{a,b}g_{a,b}\sum_{\tilde{k}=1,\tilde{k}\neq k}^{K}\alpha_{p,\tilde{k}}+\sigma^2_a}.
\end{align}

In the second phase, the GU signals are  successfully decoded by the UAV and forwarded to the GUs. Therefore, the received signal at the $n_{k}$-th port of $k$-th GU is defined as
\begin{align}
y_{k,n_{k}} = \sqrt{L_{a,k}P_a}h_{a,k}^{n_k}x+z_{k}, 
\end{align}
where $P_a$ represents the power transmitted by the UAV and $z_{k}$ denotes the AWGN with zero mean and variance of $\sigma^2_{k}$. 

Therefore, the received SINR used for decoding the common message at the $n_k$-th port of GU $k$ is given by
\begin{align}\label{eq-kc}
\gamma_{k}^{c,k} =\frac{P_a\alpha_cL_{a,k}g_{\mathrm{fas},k}}{P_bL_{a,k}g_{\mathrm{fas},k}\sum_{k=1}^{K}\alpha_{p,k}+\sigma^2_{k}}. 
\end{align}
After performing SIC and removing $s_c$ from the received signal, the SINR for the $k$-th GU to decode $s_{p,k}$ is given as
\begin{align}\label{eq-kp}
\gamma_{k}^{p,k} =\frac{P_a\alpha_{p,k}L_{a,k}g_{\mathrm{fas},k}}{P_aL_{a,k}g_{\mathrm{fas},k}\sum_{\tilde{k}=1,\tilde{k}\neq k}^{K}\alpha_{p,\tilde{k}}+\sigma^2_{k}}.
\end{align}

\section{OP Analysis}
In this letter, we define OP as the probability that the instantaneous SINR $\gamma$ falls below a specified threshold $\gamma_\mathrm{th}$, i.e., $P_\mathrm{o}=\Pr\left(\gamma\leq \gamma_\mathrm{th}\right)$. For the considered system model, the communication from the BS to the GUs is successful if and only if the transmission links from the BS to the UAV and that from the UAV to the GUs succeed. In other words, the UAV and GUs successfully decodes both the common and private messages. 
Thus, the OP of the $k$-th GU can be defined as
\begin{multline}\label{eq-op-def}
P_{\mathrm{o},k} =1-\\
\Pr\left(\gamma_{a}^{c,k}>\gamma_{\mathrm{th}}^{c,k},\gamma_{a}^{p,k}>\gamma_{\mathrm{th}}^{p,k},
\gamma_{k}^{c,k}>\gamma_{\mathrm{th}}^{c,k},\gamma_{k}^{p,k}>\gamma_{\mathrm{th}}^{p,k}\right),
\end{multline}
where $\gamma_{\mathrm{th}}^{c,k}$ and $\gamma_{\mathrm{th}}^{p,k}$ are the SINR threshold for the common message and private message, respectively. 

\begin{Proposition}
The OP of the $k$-th GU for the considered UAV-relay aided FAS using RSMA is given by
\begin{multline}
P_{\mathrm{o},k} = 1 - \left[1-\frac{\Upsilon\left(m_b,\frac{m_b}{\Omega_b}\hat{\zeta}\right)}{\Gamma\left(m_b\right)}\right]\\
\quad\times \Bigg[1-T_{\nu_{k},\mathbf{\Sigma}_{k}} \Bigg( t_{\nu_{k}}^{-1} 
\left( \frac{\Upsilon\left(m_{k},\frac{m_k}{\Omega_k}\tilde{\zeta}\right)}
{\Gamma\left(m_k\right)} \right), \dots, \\  
t_{\nu_k}^{-1} \left( \frac{\Upsilon\left(m_k,\frac{m_k}{\Omega_k}\tilde{\zeta}\right)}
{\Gamma\left(m_k\right)} \right) ; \nu_k, \theta_k \Bigg)\Bigg],
\end{multline}
where $\hat{\zeta}=\max\left\{\hat{\gamma}_{\mathrm{th}}^{c,k},\hat{\gamma}_{\mathrm{th}}^{p,k}\right\}$ and $\tilde{\zeta}=\max\left\{\tilde{\gamma}_{\mathrm{th}}^{c,k},\tilde{\gamma}_{\mathrm{th}}^{p,k}\right\}$, in which we have
\begin{align}
\hat{\gamma}_{\mathrm{th}}^{c,k}
&= \frac{\gamma_{\mathrm{th}}^{c,k} \sigma^2_a}{P_bL_{a,b} \left(\alpha_c - \gamma_{\mathrm{th}}^{c,k} \sum_{k=1}^{K}\alpha_{p,k}\right)}, \label{eq-t1}\\
\hat{\gamma}_{\mathrm{th}}^{p,k}&= \frac{\gamma_{\mathrm{th}}^{p,k} \sigma^2_a}{P_bL_{a,b} \left(\alpha_{p,k} - \gamma_{\mathrm{th}}^{p,k} \sum_{\tilde{k}=1\atop \tilde{k}\neq k}^{K}\alpha_{p,\tilde{k}}\right)},\\
\tilde{\gamma}_{\mathrm{th}}^{c,k}&= \frac{\gamma_{\mathrm{th}}^{c,k} \sigma^2_k}{P_aL_{a,k} \left(\alpha_c - \gamma_{\mathrm{th}}^{c,k} \sum_{k=1}^{K}\alpha_{p,k}\right)},\\
\tilde{\gamma}_\mathrm{th}^{p,k}&=\frac{\gamma_{\mathrm{th}}^{p,k} \sigma^2_a}{P_bL_{a,k} \left(\alpha_{p,k} - \gamma_{\mathrm{th}}^{p,k} \sum_{\tilde{k}=1\atop \tilde{k}\neq k}^{K}\alpha_{p,\tilde{k}}\right)}.\label{eq-t4}
\end{align} 
\end{Proposition}

\begin{proof}
Using \eqref{eq-ac}, \eqref{eq-ap}, \eqref{eq-kc}, and \eqref{eq-kp}, the OP in \eqref{eq-op-def} can be written as
\begin{align}\nonumber
P_{\mathrm{o},k} &= 1- \Pr\left(g_{a,b}>\max\left\{\hat{\gamma}_{\mathrm{th}}^{c,k},\hat{\gamma}_{\mathrm{th}}^{p,k}\right\}\right)\\
&\quad\times
\Pr\left(g_{\mathrm{fas},k}>\max\left\{\tilde{\gamma}_{\mathrm{th}}^{c,k},\tilde{\gamma}_{\mathrm{th}}^{p,k}\right\}\right)\\
&=1-\left[1-F_{g_{a,b}}\left(\hat{\zeta}\right)\right]\left[1-F_{{\mathrm{fas},k}}\left(\tilde{\zeta}\right)\right], \label{eq-proof1}
\end{align}
in which $F_{g_{a,b}}(\hat{\zeta})$ and $F_{{\mathrm{fas},k}}(\tilde{\zeta})$ are given by \eqref{eq-cdf1} and \eqref{eq-cdf2}, respectively. This has completed the proof.
\end{proof}

\begin{remark}\label{remark1} The power allocation factors in RSMA, namely $\alpha_c$ and $\alpha_{p,u}$, are constrained by the outage thresholds $\gamma_\mathrm{th}^{c,k}$ and $\gamma_\mathrm{th}^{p,k}$. Specifically, these thresholds must satisfy 
\begin{align} 
\gamma_{\mathrm{th}}^{c,k} &< \frac{\alpha_{c}}{\sum_{k=1}^{K}\alpha_{p,k}}, \label{eq-con3} \\
\gamma_{\mathrm{th}}^{p,k} &< \frac{\alpha_{p,k}}{\sum_{\tilde{k}=1,\tilde{k}\neq k}^{K}\alpha_{p,\tilde{k}}}. \label{eq-con4} 
\end{align} 
These bounds arise from the assumption that the values of $\hat{\gamma}_{\mathrm{th}}^{c,k}>0$,  $\hat{\gamma}_{\mathrm{th}}^{p,k}>0$, $\tilde{\gamma}_{\mathrm{th}}^{c,k}>0$, and  $\tilde{\gamma}_{\mathrm{th}}^{p,k}>0$ hold true, as shown in \eqref{eq-t1}--\eqref{eq-t4}. If the above conditions are violated, the system operates in an infeasible region, as discussed further in Section \ref{sec-num} and illustrated in Fig.~\ref{fig:1}. 
\end{remark}

\begin{remark}\label{remark2}
In the high SNR regime, the marginal CDF in \eqref{eq-cdf1} can be written as 
\begin{align}\label{eq-cdf-inf}
F_{g_{a,i}}^\infty\left(g\right) \simeq \frac{1}{m_i\Gamma\left(m_i\right)}\left(\frac{m_i}{\Omega_i}g\right)^{m_i},
\end{align}
where, by applying \eqref{eq-cdf-inf} into \eqref{eq-proof1}, the asymptotic OP is given by  
\begin{multline}
\hspace{-4mm} P_{\mathrm{o},k} = 1 - \left[1-\frac{\left(\frac{m_b}{\Omega_b}\hat{\zeta}\right)^{m_b}}{m_b\Gamma\left(m_b\right)}\right]
\Bigg[1-T_{\nu_{k},\mathbf{\Sigma}_{k}} \Bigg( t_{\nu_{k}}^{-1} 
\left(\frac{\left(\frac{m_k}{\Omega_k}\tilde{\zeta}\right)^{m_k}}{m_k\Gamma\left(m_k\right)} \right),\\
\dots,  
t_{\nu_k}^{-1} 	\left(\frac{\left(\frac{m_k}{\Omega_k}\tilde{\zeta}\right)^{m_k}}{m_k\Gamma\left(m_k\right)} \right) ; \nu_k, \theta_k \Bigg)\Bigg].
\end{multline}
\end{remark}


\section{Numerical Results}\label{sec-num}
Here, the OP performance is evaluated, which is validated through Monte Carlo simulations. Unless otherwise specified, we use the following parameter settings. For simplicity in the simulations, we consider a two-user scenario, i.e., $K=2$, where the UAV, the first GU, and the second GU are located at $\mathbf{r}_a=\left(10,10,100\right)$ m, $\mathbf{r}_1=\left(200,200,0\right)$ m, $\mathbf{r}_2=\left(180,180,0\right)$  m, respectively. Additionally, we have set $\alpha_c=0.6$, while the power allocation factors for the private messages are given by $\alpha_{p,1}=0.75\left(1-\alpha_c\right)$ and $\alpha_{p,1}=0.25\left(1-\alpha_c\right)$, respectively. Besides, $N_k=4$, $W_k=1\lambda^2$, $\Omega_i=1$, $m_b=4$, $m_k=2$, $\sigma^2_i=-70$ dBm, $P_b=P_a=5$ dBm, $\nu_k=25$,  $\mu_1=5.0188$, $\mu_2=0.3511$, and $\beta=2$. Moreover, at the reference distance of $1$ m and the carrier frequency $f_c=3.5$ GHz, we have $\eta_1=\eta_2=4.65\times 10^{-5}$.

Fig.~\ref{fig_nw} shows the performance of OP against the transmit power, where $P_a = P_b = P$, for different values of $N_k$ and $W_k$. We can observe that the OP decreases as $P$ increases, indicating a trend of improved reliability at higher power levels. Increasing the number of fluid antenna ports $N_k$ also enhances spatial diversity, while a larger fluid antenna size $W_k$ increases spatial separation, leading to a lower OP compared to deploying FPA at GUs. Additionally, it is evident that the asymptotic results closely align with the OP curves in the high SNR regime, confirming the accuracy of our analysis. 

Fig.~\ref{fig_m} studies the OP behavior in terms of $P$ for different values of the fading parameter $m_k$ when $m_b$ is fixed. As $m_k$ increases, the OP performance improves significantly due to reduced fading severity, which results in a more reliable communication link. Additionally, it can be observed that while more power is allocated to the first GU based on the RSMA scheme, the second GU consistently outperforms the first. This occurs because the second GU is located closer to the UAV and therefore benefits from better channel conditions, such as lower path-loss and reduced interference.

\begin{figure*}[h]
\centering
\begin{subfigure}{0.38\linewidth}
\centering
\includegraphics[width=\linewidth]{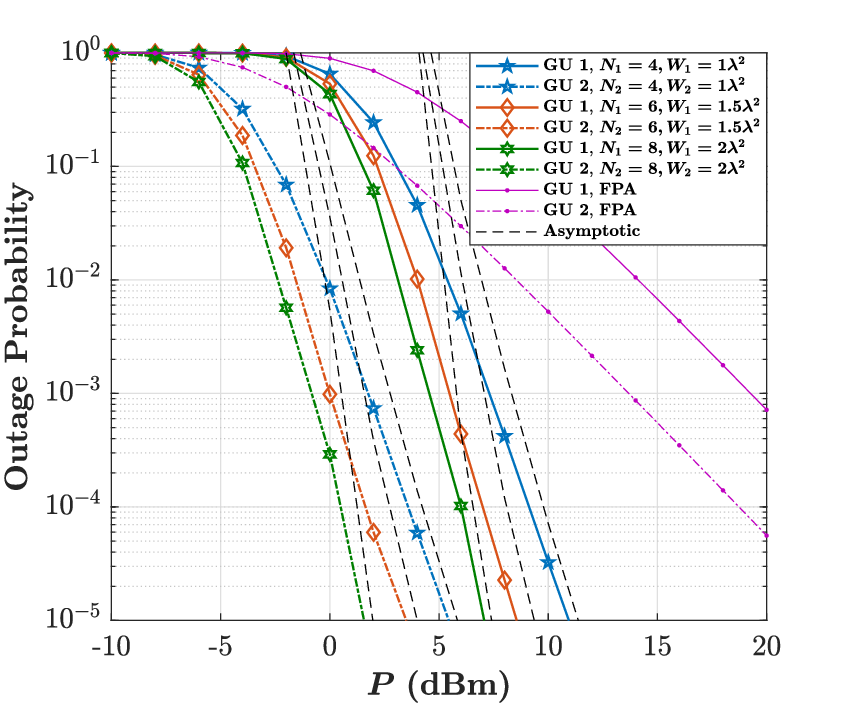}
\caption{}\label{fig_nw}
\end{subfigure}
\begin{subfigure}{0.38\linewidth}
\centering
\includegraphics[width=\linewidth]{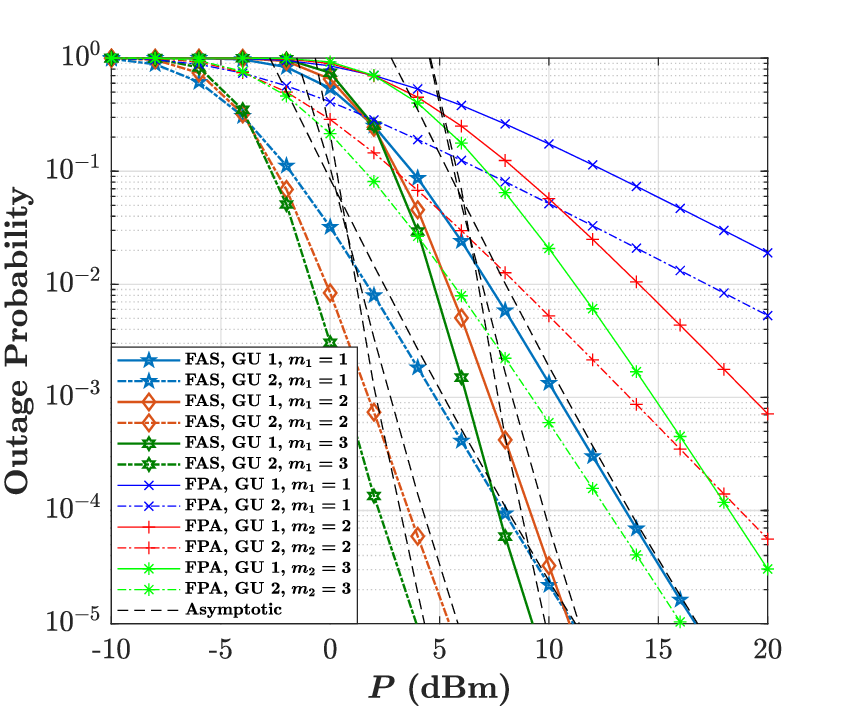}
\caption{}\label{fig_m}
\end{subfigure}
\vspace{0cm}
\begin{subfigure}{0.38\linewidth}
\centering
\includegraphics[width=\linewidth]{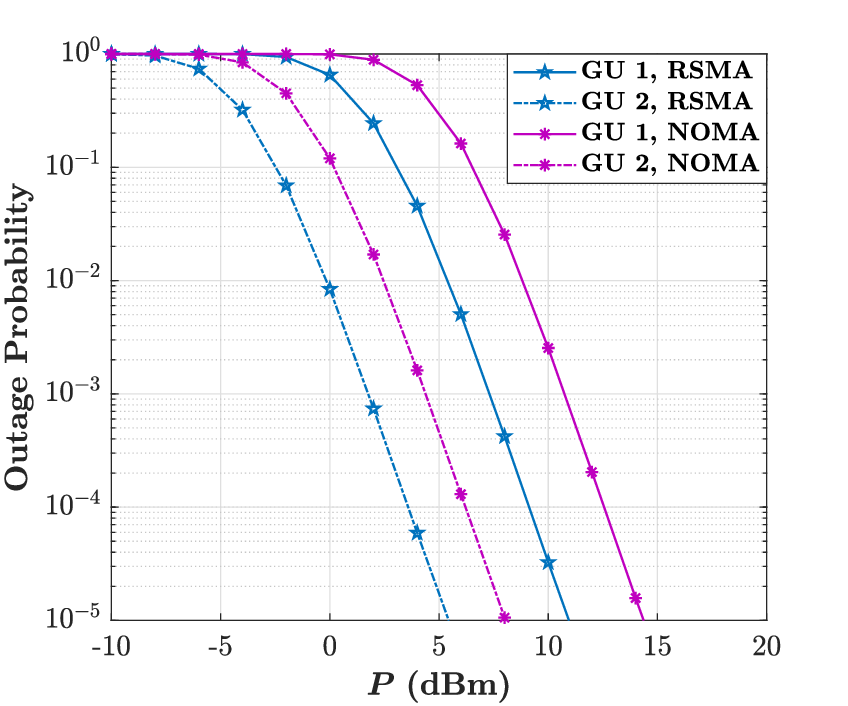}
\caption{}\label{fig_noma}
\end{subfigure}
\begin{subfigure}{0.38\linewidth}
\centering
\includegraphics[width=\linewidth]{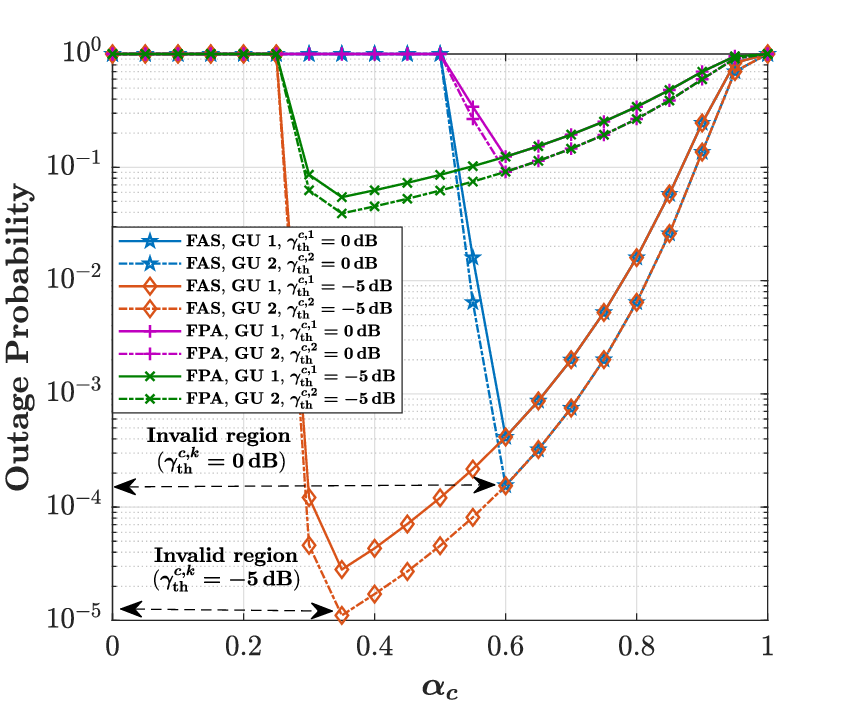}
\caption{}\label{fig_alpha}
\end{subfigure}
\caption{(a) OP versus transmit power $P$ for different values of $N_k$ and $W_k$. (b) OP versus transmit power $P$ for different values of $m_k$. (c) OP versus transmit power $P$ for different multiple access schemes. (d) OP versus RSMA common message power allocation factor $\alpha_c$ for differnt values of $\gamma_{\mathrm{th}}^{c,k}$.}\label{fig:1}
\vspace{-2mm}
\end{figure*}

Considering NOMA as a benchmark, Fig.~\ref{fig_noma} indicates the impact of RSMA on the considered UAV-relay aided FAS. We see that RSMA significantly improves the OP for both GU 1 and GU 2 compared to NOMA. The OP for both GUs under RSMA is lower across the range of transmit power levels, showing the superior performance of RSMA in reducing the likelihood of communication failure. This superiority is due to RSMA's ability to manage interference more effectively. This allows RSMA to better adapt to dynamic channel conditions and optimize resource allocation, particularly in FAS where antenna configurations can change rapidly. 

Fig.~\ref{fig_alpha} illustrates the OP results versus the RSMA common message power allocation factor $\alpha_c$ for various SNR threshold values $\gamma_{\mathrm{th}}^{c,k}$. We see that when $\alpha_c$ and $\gamma_{\mathrm{th}}^{c,k}$ fail to meet a certain condition (see Remark \ref{remark1}), the system enters an invalid region where communication becomes infeasible. Also, the OP decreases with increasing $\alpha_c$ initially, reaching a minimum at an optimal value before increasing again, creating a convex-like trend. This behavior stems from the trade-off in RSMA between allocating power to the common and private messages. In fact, when $\alpha_c$ is too large, insufficient power is left for the private message, causing its SNR to drop below the threshold and resulting in an outage. By contrast, when $\alpha_c$ is too small, the common message SNR becomes inadequate, leading to a common message outage. It is also observed that the position of the optimal $\alpha_c$ shifts depending on $\gamma_{\mathrm{th}}^{c,k}$, with higher threshold values increasing the OP and widening the invalid region as the upper limit shifts to the right.

\vspace{-2mm}
\section{Conclusion}\label{sec-con}
This letter studied how the incorporation of FAS technology can enhance the performance of UAV-aided multiuser communication networks. In particular, we focused on a scenario in which a FPA BS served multiple FAS-equipped GUs with the help of a UAV acting as an aerial relay. By employing RSMA signaling technique at the BS and a HD DF strategy for the UAV, we derived an analytical expression for the OP and analyzed its behavior in the high-SNR regime. The findings highlighted the advantages of using FAS-equipped GUs over traditional FPA setups, showing that this approach leads to a more efficient system, particularly in high-SNR conditions. Our results suggest that leveraging FAS, in conjunction with UAV support, has the potential to greatly improve the reliability and efficiency of future communication networks.

\vspace{-2mm}

\end{document}